\documentclass[10pt]{amsart}
\pdfoutput=1

\usepackage{amssymb}
\usepackage{amsthm}
\usepackage{amsfonts}
\usepackage{amsmath}
\usepackage{pmboxdraw}
\usepackage{verbatim} 
\usepackage{graphicx}
\usepackage{color}
\usepackage[colorlinks=true, citecolor=blue, filecolor=black, linkcolor=black, urlcolor=black]{hyperref}
\usepackage{cite}
\usepackage[normalem]{ulem}
\usepackage{subcaption}
\usepackage{bbm}
\usepackage{bm}
\usepackage{mathtools}
\usepackage{todonotes}
\usepackage{kantlipsum}
\allowdisplaybreaks

\usetikzlibrary{arrows.meta, positioning} \usetikzlibrary{calc}

\newcommand{\RN}[1]{%
	\textup{\uppercase\expandafter{\romannumeral#1}}%
}

\def\C{\mathbb{C}}

\def\R{\mathbb{R}}

\newcommand{\re}{\operatorname{Re}}
\newcommand{\im}{\operatorname{Im}}

\theoremstyle{plain}
\newtheorem{thm}{Theorem}[section]

\newtheorem{cor}[thm]{Corollary}
\newtheorem{lem}[thm]{Lemma}

\theoremstyle{remark}

\newtheorem{rem}{Remark}

\numberwithin{equation}{section}

\begin{document}

\title[Partition functions of lemniscate ensembles]{Anomalous free energy expansions of planar Coulomb gases: multi-component and conformal singularity 
}
\author{Sung-Soo Byun}
\address{Department of Mathematical Sciences and Research Institute of Mathematics, Seoul National University, Seoul 151-747, Republic of Korea}
\email{sungsoobyun@snu.ac.kr}

\begin{abstract}
We study the partition function  
$$
Z_n = \int_{\mathbb{C}^n } \prod_{1 \le j<k \le n} |z_{j}-z_{k}|^{2} \prod_{j=1}^{n} |z_j|^{2c}\, e^{-n V(z_{j})}\frac{d^{2}z_{j}}{\pi},
$$
where $c>-1$ and
$$
V(z)= |z|^{2d}-t(z^{d}+\overline{z}^{d}), \qquad t >0, \, d \in \mathbb{N}.
$$
The associated droplet reveals a topological phase transition: for $t > 1/\sqrt{d}$, it consists of $d$ connected components; whereas for $t < 1/\sqrt{d}$, it becomes simply connected and contains the origin, where a conformal singularity arises.
In both regimes, we establish the asymptotic expansion  
$$ 
\log Z_n = C_1 n^2 + C_2 n \log n + C_3 n + C_4 \log n + C_5 + O(n^{-1}),
$$ 
as $n \to \infty$, and derive all coefficients explicitly. In the multi-component regime $t > 1/\sqrt{d}$, the constant term $C_5$ exhibits an oscillatory behaviour that depends on the congruence class of $n$ modulo $d$. In particular, in the special case $c = 0$ with $n$ divisible by $d$, our result confirms a conjecture of Deaño and Simm. In contrast, in the conformal singularity regime $t < 1/\sqrt{d}$, the oscillatory behaviour disappears, while additional contributions in $C_4$ arise beyond the scope of the conjecture of Jancovici et al. As a special case $d = 1$, our result yields the asymptotic expansion for the moments of the characteristic polynomial of the complex Ginibre ensemble with finite exponent. In the bulk regime, we further derive the full expansion in powers of $1/N$, thereby providing a precise evaluation of the error term in the result of Webb and Wong.
\end{abstract}

\maketitle

\section{Introduction and results} 

For a given external potential $W: \C \to \R$, the two-dimensional Coulomb gas ensemble $\{z_j\}_{j=1}^n$ ($z_j \in \C$) \cite{Fo10,Se24} at inverse temperature $\beta=2$ is defined by the probability measure 
\begin{equation} \label{def of Gibbs}
\frac{1}{Z_n(W)} \prod_{1 \le j<k \le n} |z_{j}-z_{k}|^{2} \prod_{j=1}^{n}e^{-n W(z_{j})}\frac{d^{2}z_{j}}{\pi}, 
\end{equation}
where $Z_n(W)$ is the partition function.
As $n \to \infty$, the empirical measure of $\{z_j\}_{j=1}^n$ weakly converges to the equilibrium measure $\sigma_W$, the unique minimiser of the weighted logarithmic energy  
\begin{equation} \label{def of log energy}
	I_W[\mu]:= \int_{ \mathbb{C}^2 } \log \frac{1}{ |z-w| }\, d\mu(z)\, d\mu(w) +\int_{ \mathbb{C} } W(z) \,d\mu(z). 
\end{equation}
In addition, under standard assumptions from potential theory \cite{ST97}, the equilibrium measure \( \sigma_W \) takes the form
\begin{equation} 
d \sigma_W(z) =  \Delta W(z)\, \mathbbm{1}_{S_W}(z) \, \frac{d^2z}{\pi}, \qquad \Delta=\partial \bar{\partial},   
\end{equation}
where the compact support $S_W$ is called the droplet.

In recent years, there has been growing interest in the asymptotic behaviour of the free energy $\log Z_n(W)$ as $n \to \infty$. Accumulated knowledge (see \cite[Section 5.3]{BF25}, \cite[Section 9.3]{Se24} and introductions of \cite{BKS23,BSY24,ACC23,BKSY25}) now suggests that for an $n$-independent potential $W$, the partition function $Z_n(W)$ admits an asymptotic expansion of the form
\begin{align}
\begin{split}
	\log Z_n(W)  & = - I_W[\sigma_W] n^2 +\frac12 n \log n + \Big( \frac{\log(2\pi)}{2}-1 - \frac12 \int_\C \log (\Delta W) \,d\sigma_W \Big) n 
\\
&\quad +  \frac{6-\chi}{12}\log n   + \chi \, \zeta'(-1) + \frac{\log(2\pi)}{2} + \mathcal{F}[W] + \mathcal{G}_n[W] + \mathcal{H}_n[W] + o(1), \label{Z expansion}
\end{split}
\end{align} 
as $n \to \infty$, where $\chi \equiv \chi(S_W)$ is the Euler characteristic of the droplet $S_W$ and $\zeta$ is the Riemann zeta function. (See also \cite{Ch22,Ch23,BP24,AFLS25} for studies addressing cases with hard edges.) 
 
The leading three terms of the free energy expansion \eqref{Z expansion}—of orders $O(n^2)$, $O(n \log n)$, and $O(n)$—were established in \cite{LS17} for general Coulomb gases. These results were subsequently refined in \cite{Se23, BBNY19, AS21}, where quantitative error bounds were obtained.
In contrast to the first three terms, the remaining two terms of order $O(\log n)$ and $O(1)$ remain open and present significant challenges. A key distinction is that these latter terms depend in a highly non-trivial way on the topological and conformal geometric nature of the equilibrium measure. The $O(\log n)$ term was proposed in the early works of Jancovici et al. \cite{JMP94, TF99}, highlighting its dependence on the topology of the droplet (see also \cite{CFTW15}). The $O(1)$ term—particularly the constant $\mathcal{F}[W]$—is believed to involve conformal geometric functionals (the zeta regularised determinant of spectral Laplacian), as predicted by Wiegmann and Zabrodin using conformal field theory \cite{ZW06}.  

Of particular interest in this note are the last two terms, $\mathcal{G}_n$ and $\mathcal{H}_n$, which are typically zero in most of the existing literature. However, these terms in fact encode additional intrinsic information of the equilibrium measure.

\begin{itemize}
    \item \textbf{(Multi-component term $\mathcal{G}_n$)} 
The term $\mathcal{G}_n[W]$ remains bounded as $n \to \infty$ but does \textit{not} converge. This term is believed to arise only when the droplet is \textit{disconnected}. This phenomenon was studied in \cite{ACC23} (see also \cite{Ch22,Ch23,ACCL24,ACC23a}), where the authors established the expansion \eqref{Z expansion} for a class of radially symmetric potentials whose droplets form concentric annuli.
However, no concrete example has yet confirmed the presence of the oscillatory term $\mathcal{G}_n$ in the free energy expansion \eqref{Z expansion} for non-radially symmetric potentials (see however \cite{AC24} for a discussion of the oscillatory behaviour in the context of fluctuation phenomena in Coulomb gases with multi-component droplets).
    \smallskip 
    \item \textbf{(Conformal singularity term $\mathcal{H}_n$)} In order to discuss the term $\mathcal{H}_n$, we define  
\begin{equation} \label{def of conformal metric}
\varphi(z) := \frac{1}{2} \log \Delta W(z)
\end{equation}  
so that $e^{2\varphi(z)}$ serves as a conformal metric in the expression for the equilibrium measure \eqref{def of eq msr}. In much of the literature on free energy expansions beyond the \( O(n) \) term, the conformal metric factor \( \varphi \) is assumed to be ``regular'', in the sense that it is bounded above and below within the droplet. This assumption is partially motivated by the fact that if \( \varphi \) becomes singular at some point, then the associated spectral determinant of the Laplacian also diverges, influencing unbounded terms in the expansion (see Remark~\ref{Rem_diverging Laplacian}). Therefore, it is believed that the contribution \( \mathcal{H}_n  \to \infty \) arises in cases where \( \varphi \) exhibits such singularities inside the droplet.

\end{itemize}

In contrast to much of the literature on free energy expansions in two dimensions, the multi-component and conformal singularity regimes remain less explored. This note aims to contribute to this direction by considering a class of potentials exhibiting discrete rotational symmetry. In particular, we focus on a well-studied model known as the \textit{lemniscate ensemble}, which enables a detailed analysis of the non-standard terms \( \mathcal{G}_n \) and \( \mathcal{H}_n \), including explicit formulae.

\medskip 
 
By definition, the lemniscate ensemble \cite{BM15,BGM17,BEG18,BLY21,BY23} follows \eqref{def of Gibbs} with the potential 
\begin{align} \label{def of V lemniscate}
V_{d,t}(z) := |z|^{2d}-t(z^{d}+\overline{z}^{d}), \qquad t \ge 0, \, d \in \mathbb{N}. 
\end{align}
Using the discrete rotational symmetry $V_{d,t}(z)=V_{d,t}(e^{2\pi i/d} z)$, it was shown in \cite{BM15} that the associated equilibrium measure is 
\begin{equation} \label{def of eq msr}
d\sigma_V(z) = d^2|z|^{2d-2} \mathbbm{1}_{S_V}(z) \frac{d^2z}{\pi}, \qquad S_V : = \Big \{  z \in \C : (\re (z^d)-t)^2+(\im (z^d))^2 \le \frac{1}{d} \Big\}.
\end{equation}
The droplet $S_V$ undergoes a topological phase transition at the critical value 
\begin{equation}
t_c := 1/\sqrt{d}. 
\end{equation}
For $t < t_c$, the droplet consists of a single connected component, whereas for $t > t_c$, it splits into $d$ disjoint components, see Figure~\ref{Fig_SV}.

\begin{figure}[t]
	\begin{subfigure}{0.24\textwidth}
	\begin{center}	
		\includegraphics[width=\textwidth]{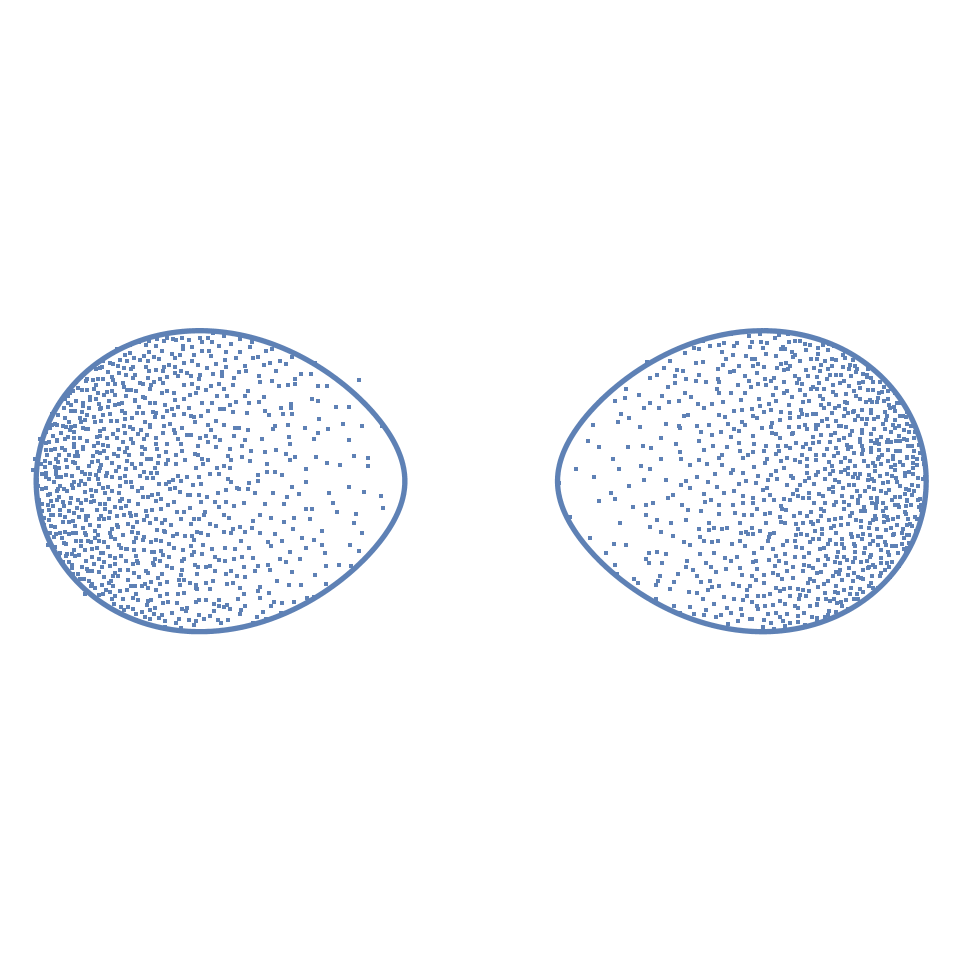}
	\end{center}
	\subcaption{$d=2, t=0.75$}
\end{subfigure}	
\begin{subfigure}{0.24\textwidth}
	\begin{center}	
		\includegraphics[width=\textwidth]{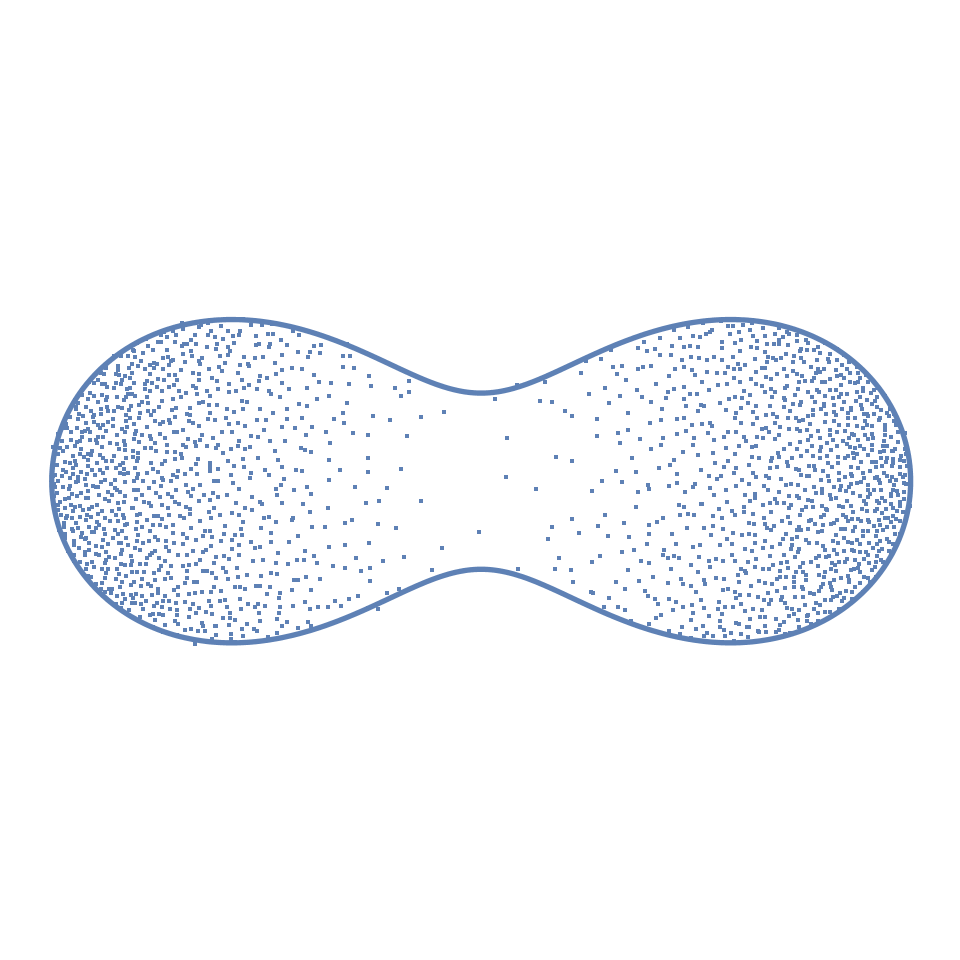}
	\end{center}
	\subcaption{$d=2, t=0.65 $}
\end{subfigure}	
\begin{subfigure}{0.24\textwidth}
	\begin{center}	
		\includegraphics[width=\textwidth]{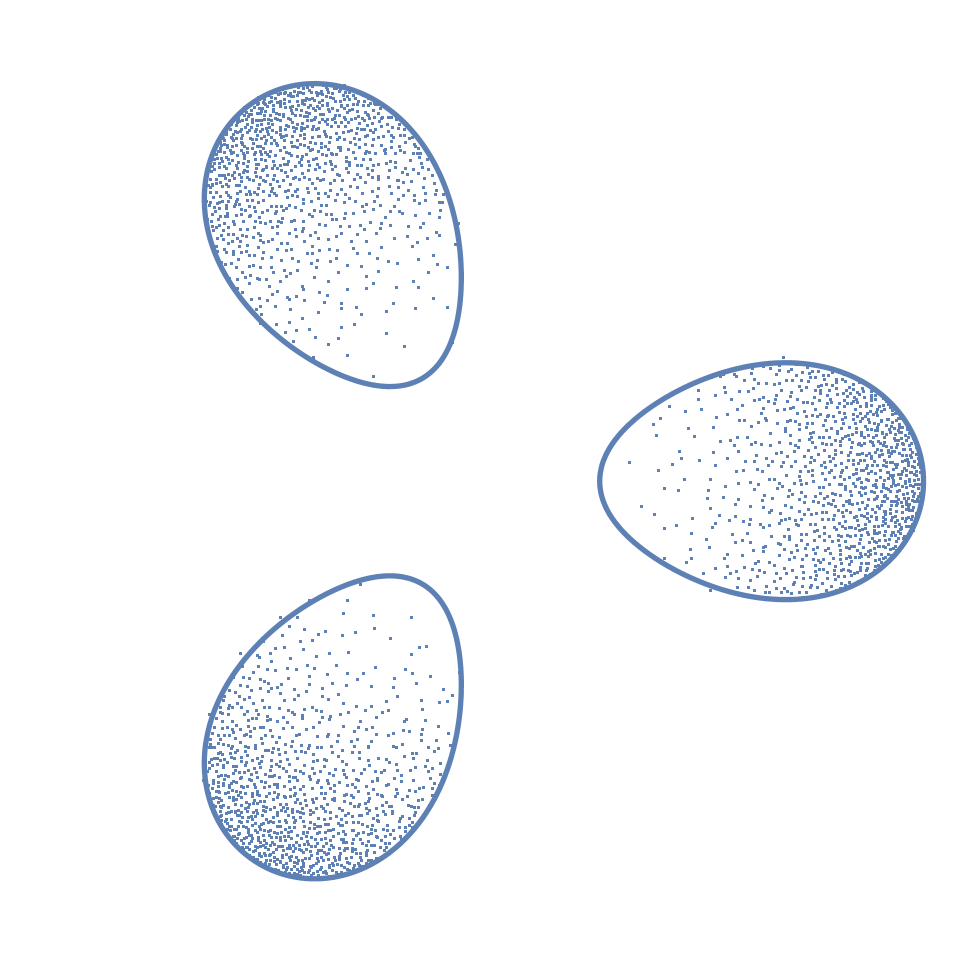}
	\end{center}
	\subcaption{$d=3, t=0.6 $}
\end{subfigure}	  
\begin{subfigure}{0.24\textwidth}
	\begin{center}	
		\includegraphics[width=\textwidth]{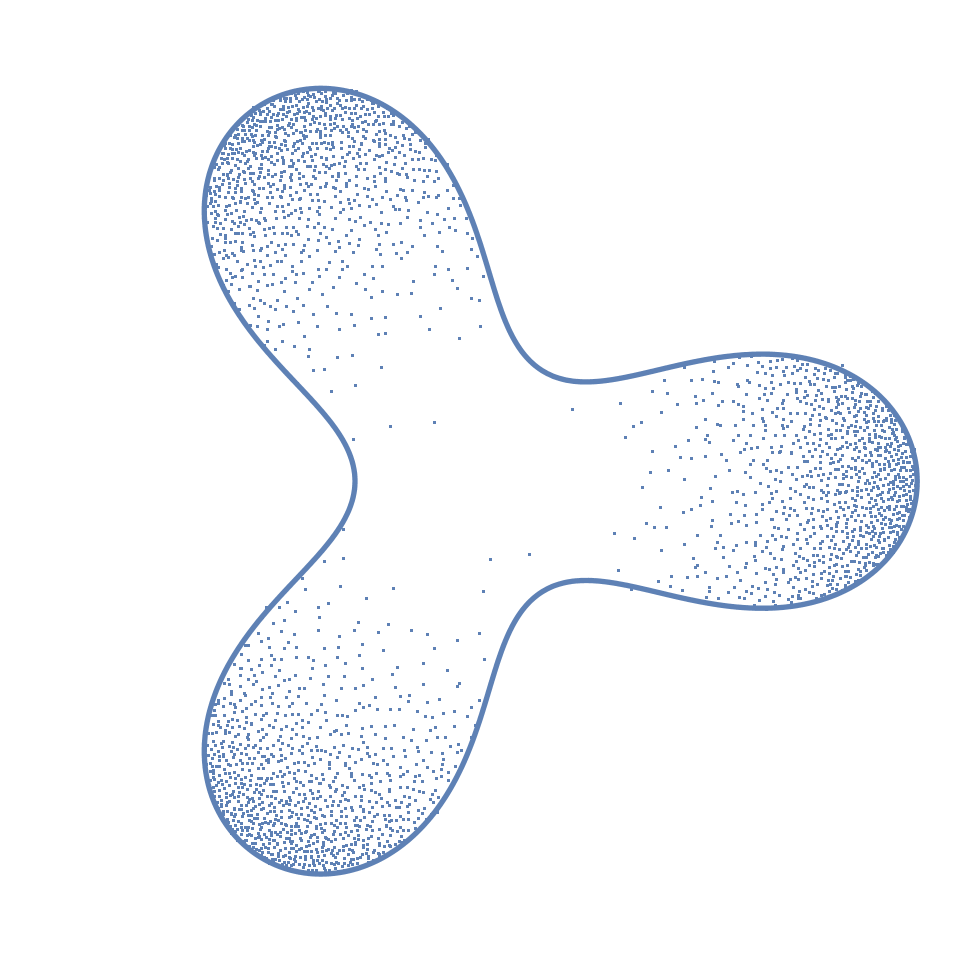}
	\end{center}
	\subcaption{$d=3, t=0.55 $}
\end{subfigure}	
	\caption{Illustration of the lemniscate ensembles, where $N=1000d$.} \label{Fig_SV}
\end{figure}

In addition to \eqref{def of V lemniscate}, for $n \in \mathbb{N},$ we define 
\begin{equation} \label{def of Vc lemniscate}
V_{d,t}^{(c)}(z):= V_{d,t}(z)-\frac{2c}{n} \log |z|, \qquad c > -1. 
\end{equation}
Here, the $\log|z|$ term represents the insertion of a point charge at the origin (see e.g. \cite{AKS23}), often referred to as a spectral singularity in random matrix theory. This insertion is intimately connected to the characteristic polynomials of random matrices \cite{AV03}, and the associated planar orthogonal polynomials have been extensively studied in the literature~\cite{BBLM15,BGM17,BKP23,BEG18,BFKL25,BSY24,BY23,LY17,LY19,LY23}.
In our main result, we derive the asymptotic expansion of the partition functions associated with the potential \eqref{def of Vc lemniscate}. 

\begin{thm}[\textbf{Free energy expansion of lemniscate ensembles}] \label{Thm_free energy expansion}
Let $V_{d,t}^{(c)}$ be the external potential defined by \eqref{def of Vc lemniscate}. Suppose that $d \in \mathbb{N}$, $c > -1$ and $t \ge 0$. 
Then as $n \to \infty$, 
\begin{align}
\begin{split}  \label{log Z expansion main}
\log Z_n(V_{d,t}^{(c)} )  = C_1 \, n^2 + C_2 \, n \log n + C_3 \, n + C_4\, \log n + C_5 +O(n^{-1}),
\end{split}
\end{align} 
where 
\begin{equation}\label{def of C1 C2}
C_1= t^2 -\frac{3}{4d}-\frac{\log d}{2 d},   \qquad C_2= \frac12, 
\end{equation}
and $C_3$, $C_4$, and $C_5$ are given as follows. 
\begin{itemize} 
\item \textup{\textbf{(Multi-component regime)}} Suppose $t>1/\sqrt{d}$. Then  
\begin{align}
C_3 &=  \frac{\log(2\pi)}{2}-1 + \frac{ 1+2c-d   }{ d } \log t - \log d ,
\\
C_4 &= \frac{6-d}{12} , 
\\ 
\begin{split}
C_5 &= d \,\zeta'(-1)  + \frac{\log(2\pi)}{2}  + \frac{1}{d} \Big( c(d-c-1) -\frac{(d-1)(2d-1)}{6} \Big) 
\log\Big( \frac{dt^2-1}{dt^2} \Big)  
\\
&\quad + \frac{d}{12} \log d  +  d \, \Big\{ \frac{n}{d} \Big\} \Big( \Big\{ \frac{n}{d} \Big\} -1 \Big)  \log( \sqrt{d}t ). 
\end{split}
\end{align}
Here,  $\{ n/d \} = n/d-\lfloor n/d \rfloor$ is the fractional part of $n/d$. 
\smallskip 
\item \textup{\textbf{(Conformal singularity regime)}} Suppose $t<1/\sqrt{d}$. Then  
\begin{align}
C_3  &=  \frac{\log(2\pi)}{2}-1 +\frac{1+2c-d}{2d}(dt^2-1) -\frac{  1+2c+d }{2d} \log d ,
\\
C_4 &= \frac{5}{12} + \frac{(d-1)^2}{12d} -\frac{c(d-c-1)}{2d}  , 
\\ 
\begin{split}
C_5 &= d \,\zeta'(-1)  + \frac{3+2c-d}{4}\log(2\pi) 
\\
&\quad + \Big( \frac{d}{12} - \frac{(d-1)(2d-1)}{12d} + \frac{c(d-c-1)}{2d} \Big) \log d - \sum_{\ell=0}^{d-1} \log G\Big(  \frac{\ell+1+c}{d} \Big).
\end{split}
\end{align} 
Here, $G$ is the Barnes $G$-function. 
\end{itemize}
\end{thm}

\begin{rem}[Comparison with Deaño--Simm conjecture \cite{DS22}]
In \cite{DS22}, an asymptotic formula for \( Z_n(V_{d,t}^{(c)}) \) was proposed in the special case \( n = dN \) and \( c = 0 \), based on the fluctuation theory \cite{AHM11} for general random normal matrix ensembles. Here, we compare our result with their prediction. Following \cite{DS22}, write $Z_{n}^{ \mathrm{Lem}_{d} }(t) := Z_n( V_{d,t}  )$.
Then it was proposed in \cite[Eq.~(3.46) and Appendix A]{DS22}\footnote{There appears to be a minor sign error in \cite[Eq.~(3.46)]{DS22}. In that expression, the exponent of the $(t/t_c)$ term in \eqref{asymp 1 corrected} is given as $+N(d-1)$. However, according to the derivation provided in \cite[Appendix~A]{DS22}, the exponent should follow from  
$$
\sum_{ \ell=0 }^{d-1} N \gamma_\ell = -N(d-1) ,\qquad \gamma_{\ell} = -2 \Big( 1-\frac{\ell+1}{d} \Big) , 
$$
see \cite[Eq.~(A.7)]{DS22}.} that for $t>t_c=1/\sqrt{d}$, 
\begin{align} \label{asymp 1 corrected}
	 Z_{dN}^{\mathrm{Lem}_{d}}(t)  = \big(Z_{N}^{\mathrm{Lem}_{1}}(t\sqrt{d})\big)^{d}\,  c_{N,d}\Big( \frac{t}{t_{c}} \Big)^{-N(d-1)}\bigg( 1-\Big( \frac{t_{c}}{t} \Big)^{2} \bigg)^{-\kappa_{d}} (1+o(1)),  
\end{align}
as $N \to \infty$, where  
\begin{equation} \label{def of kappa d}
c_{N,d} = \frac{(Nd)!}{ (N!)^d } \frac{1}{d^{N(dN+2d+1)/2}}  , \qquad \kappa_{d}:=\frac{(d-1)(2d-1)}{6d}. 
\end{equation} 
Note that by definition, 
\begin{equation} 
Z_{N}^{\mathrm{Lem}_{1}}(t\sqrt{d}) =  e^{d(Nt)^{2}}  Z_{N}^{\mathrm{Gin}},  
\end{equation} 
where $Z_N^{ \rm Gin } $ is the partition function of the Ginibre ensemble, see \eqref{def of ZN Ginibre}. 
By using the Stirling's formula \eqref{gamma asymp}, direct computation gives
\begin{align}\label{asymp 2}
\log c_{N,d} = -\frac{d \log d}{2}N^{2} - \frac{\log d}{2}N - \frac{d-1}{2}\log N + \frac{\log d - (d-1)\log(2\pi)}{2} +  O(N^{-1}) , 
\end{align}
as $N\to\infty$. Then by using the asymptotic behaviour \eqref{Barnes G asymp} of the Barnes $G$-function, after straightforward computations, one can check that \eqref{asymp 1 corrected} is consistent with the formula \eqref{log Z expansion main} for the special case $n=dN$ and $c=0$.
\end{rem}

\begin{rem}[Rotationally symmetric case]
For $t = 0$, the potential \eqref{def of Vc lemniscate} becomes rotationally symmetric. In this case, the monomials form an orthogonal system with respect to the rotationally symmetric weight. It follows from \eqref{Zn in terms of det} and the integral evaluation 
$$
\int_{ \C } |z|^{2j} e^{-n V_{d,0}^{(c)}(z)} \frac{d^2z}{\pi} =  \int_{ \C } |z|^{2j+2c} e^{ -n |z|^{2d} }  \frac{d^2z}{\pi}  = \frac{ \Gamma( \frac{j+c+1}{d} ) }{ d n^{ \frac{j+c+1}{d} } } 
$$
that 
\begin{align}
Z_n(V_{d,0}^{(c)})  = \frac{n!}{d^n \, n^{ \frac{n(n+2c+1)}{2d} }  } \prod_{j=0}^{n-1} \Gamma(\tfrac{j+c+1}{d}). 
\end{align}
Then with $n=dN+m$, by using 
\begin{equation} \label{Barnes G def}
G(z+1)=\Gamma(z)G(z),\qquad G(1)=1, 
\end{equation}
we have 
\begin{align}
\begin{split}  \label{Zn t0 Barnes G}
Z_{n}(V_{d,0}^{(c)}) & =  \frac{n!}{d^n \, n^{ \frac{n(n+2c+1)}{2d} }  }  \prod_{\ell=0}^{d-1} \prod_{k=0}^{N-1} \Gamma(k+\tfrac{\ell+c+1}{d}) \prod_{\ell=0}^{m-1} \Gamma(N+\tfrac{\ell+c+1}{d}) 
\\ 
& = \frac{n!}{d^n \, n^{ \frac{n(n+2c+1)}{2d} }  }   \prod_{\ell=0}^{d-1} \frac{ 1 }{ G( \frac{\ell+c+1}{d} ) }  \prod_{\ell=0}^{m-1} G(N+1+\tfrac{\ell+c+1}{d}) \prod_{\ell=m}^{d-1} G(N+\tfrac{\ell+c+1}{d}). 
\end{split}
\end{align} 
Combining \eqref{Zn t0 Barnes G} and the asymptotic behaviour \eqref{Barnes G asymp} of the Barnes $G$-function, one can directly check that \eqref{log Z expansion main} holds in the special case $t = 0$.
\end{rem}

\begin{rem}[Oscillatory term] In Hermitian random matrix theory, it is well known that in the multi-cut regime, where the limiting spectrum consists of multiple disconnected intervals, the oscillatory terms in the free energy expansion can be expressed in terms of the Riemann theta function (see e.g. \cite{CFWW25,CGM15,BG24} and references therein). In particular, simplifications such as quasi-periodicity occur when the equilibrium measure distributes equal mass across each component of the support.

A description based on theta functions also arises in the study of non-Hermitian random matrix models under rotational symmetry, see \cite{ACC23,Ch22,Ch23,ACCL24,ACC23a}.
 In our current setting—the lemniscate ensembles—the enhanced symmetries lead to drastic simplifications of the oscillatory terms, reducing them essentially to the elementary form \( x(x - 1) \), where \( x = \{n/d\} \). Nonetheless, one can observe structural similarities with the general theta-function structure in previous works. 
More precisely, in the two-dimensional setting, the oscillatory contribution has been computed in the radially symmetric case, where it admits an explicit expression in terms of a theta-like function:
\begin{equation} 
\Theta(x; \rho, a) = x(x - 1)\log(\rho) + x\log(a) 
+ \sum_{j=0}^{\infty} \log\Big(1 + a \, \rho^{2(j + x)}\Big) 
+ \sum_{j=0}^{\infty} \log\Big(1 + a^{-1} \, \rho^{2(j + 1 - x)}\Big),
\end{equation} 
see \cite[Eq.~(3.77)]{Ch23}. This comparison motivates the conjecture that the lemniscate ensemble corresponds to a degenerate case in which only the leading \( x(x - 1) \) term survives. At present, no general result confirms this degeneracy, and developing a full theoretical understanding remains an open problem for future investigation.
\end{rem}

\begin{rem}[Full asymptotic expansion in powers of $1/N$]
For the regime where $t < 1/\sqrt{d}$, the error term $O(n^{-1})$ in \eqref{log Z expansion main} can indeed be explicitly computed. To be more concrete, the full asymptotic expansion in powers of $1/n$ can be explicitly computed and expressed in terms of the Bernoulli number. This aspect will be discussed in more detail for the $d=1$ case, see Theorem~\ref{Prop_moment GinUE} and Remark~\ref{Rem_full expansion}.  
\end{rem}

We note that, due to the $n$-dependence of the potential in \eqref{def of Vc lemniscate}, Theorem~\ref{Thm_free energy expansion} does not fall within the class of cases where the general asymptotic formula \eqref{Z expansion} applies. However, in the special case $c=0$, one can verify the conjectural asymptotic form \eqref{Z expansion}. As a consequence of Theorem~\ref{Thm_free energy expansion}, we have the following.

\begin{cor} \label{Cor_free energy}
Let $W=V_{d,t}$ with $d \in \mathbb{N}$. Then the free energy expansion \eqref{Z expansion} holds with the equilibrium measure \eqref{def of eq msr}, where $\mathcal{F}$, $\mathcal{G}_n $ and $\mathcal{H}_n$ are given as follows.
\begin{itemize}
    \item For $t > 1/\sqrt{d}$, we have $\chi=d$, $\mathcal{H}_n[W] = 0$ and 
\begin{align}
\mathcal{F}[W] &=     -\frac{(d-1)(2d-1)}{6d}  \log\Big( \frac{dt^2-1}{dt^2} \Big)  + \frac{d}{12} \log d , \\
\mathcal{G}_n[W]  & =  d \, \Big\{ \frac{n}{d} \Big\} \Big( \Big\{ \frac{n}{d} \Big\} -1 \Big)  \log( \sqrt{d}t ). 
\end{align}  
\item For $t< 1/\sqrt{d}$, we have $\chi=1$, $\mathcal{G}_n[W] = 0$ and 
\begin{align}
\mathcal{F}[W] &= \Big( \frac{d}{12} - \frac{(d-1)(2d-1)}{12d}  \Big) \log d , 
\\
\mathcal{H}_n[W] & = \frac{(d-1)^2}{12d}\log n + (d-1)\Big( \zeta'(-1)-\frac{\log(2\pi)}{4} \Big) - \sum_{\ell=0}^{d-1} \log G\Big(  \frac{\ell+1}{d} \Big).  \label{def of mathcal Hn bulk}
\end{align} 
\end{itemize} 
\end{cor}

\begin{rem}[Divergence of the spectral Laplacian in the conformal singularity] \label{Rem_diverging Laplacian}
In the conformal singularity regime, the most notable feature is the appearance of an additional $O(\log n)$ term in \eqref{def of mathcal Hn bulk}, which lies beyond the scope of the original conjecture by Jancovici et al. We now outline a formal argument indicating how such a term naturally arises.

As mentioned earlier, the $O(1)$-term in the free energy expansion is believed to correspond to the zeta-regularised determinant of the spectral Laplacian~\cite{ZW06}. However, in the presence of conformal singularity—as is the case here—this determinant diverges. More precisely, the $O(1)$-term in~\cite{ZW06} includes a contribution of the form \cite[Eqs.~(4.8), (4.9), (5.27)]{ZW06}
\begin{align} \label{eqn for diverging Laplacian}
\frac{1}{12} \int_{ S_V } | \nabla \varphi |^2 \, \frac{d^2z}{\pi}, 
\end{align}
where $\varphi$ is given by \eqref{def of conformal metric}. (See also~\cite[Eq.~(C.7)]{ZW06} and~\cite[Eq.~(1.31)]{BKS23} for explicit expressions in the rotationally symmetric setting.)

In our case, the conformal factor $\varphi$ satisfies $\nabla \varphi(z) = (d-1)/z$. Since $S_V$ contains the origin when $t < 1/\sqrt{d}$, the integral \eqref{eqn for diverging Laplacian} diverges due to the singularity at the origin.
Nevertheless, the natural regularisation scale is the exclusion of a neighbourhood of radius $O(n^{-1/(2d)})$ around the origin, which corresponds to the microscopic scaling of the potential \eqref{def of V lemniscate}, see e.g.~\cite{AKS23}. Thus, we obtain the regularised form
$$
\frac{1}{12} \int_{ S_V \cap \{ |z| >c n^{ -\frac{1}{2d} } \} } | \nabla \varphi(z) |^2 \, \frac{d^2z}{\pi} = \frac{(d-1)^2}{12d} \log n +O(1), 
$$
which matches the additional $O(\log n)$ term appearing in \eqref{def of mathcal Hn bulk}.
We refer to \cite[Eq.~(1.35)]{BKS23} for a similar discussion concerning the appearance of such logarithmic terms in the presence of diverging spectral determinant. (We also refer the reader to \cite{JV23, Kl17} for related models in which the coefficient of the $\log n$ term depends on the geometric singularities of the system.)
\end{rem}

\begin{rem}[Vanishing or diverging densities of the equilibrium measure]
Continuing the discussion from the previous remark, observe that the function
\begin{equation}
d \mapsto \frac{(d-1)^2}{12d} = \frac{1}{12}\Big( d+\frac{1}{d}\Big)-\frac{1}{6}, 
\end{equation}
which appears in the $O(\log n)$ term of \eqref{def of mathcal Hn bulk}, is invariant under the transformation $d \mapsto 1/d$. This suggests that a vanishing equilibrium density of order \( O(|z|^{2d-2}) \) has the same effect on the \( O(\log n) \) term in the free energy expansion as a diverging density of order \( O(|z|^{2/d - 2}) \).

Although the class of potentials considered in \eqref{def of V lemniscate} does not include cases with diverging equilibrium densities, we may still illustrate this dual behaviour through a simple yet instructive example of the potential:
\begin{equation}
\widetilde{V}_d(z) := |z|^{ 2/d }, \qquad d \in \mathbb{N}. 
\end{equation} 
By using \eqref{Zn in terms of det}, \eqref{Barnes G def} and the multiplication formula of the gamma function \cite[Eq.~(5.5.6)]{NIST}, we have
\begin{equation}
Z_n(\widetilde{V}_d)  = n!  \frac{ d^{ dn(n+1)/2+n/2 } }{ n^{ dn(n+1)/2 } }  (2\pi)^{n (1-d)/2 }   \prod_{ \ell=0 }^{d-1}  \frac{ G(n+1+\frac{\ell}{d}) }{ G(1+\frac{\ell}{d}) } .
\end{equation} 
Then it follows from \eqref{Barnes G asymp} that   
\begin{align}
\begin{split}
\log Z_n(\widetilde{V}_d) & = \Big( \frac{d\log d}{2}-\frac{3d}{4} \Big) n^2+\frac12 n \log n  +\Big( \frac{\log(2\pi)}{2}-1 +\frac{d+1}{2}\log d- \frac{d-1}{2} \Big)n 
\\
& \quad + \Big( \frac{5}{12} + \frac{(d-1)^2}{12d}  \Big) \log n + d \zeta'(-1)+\frac{d+1}{4} \log (2\pi)- \sum_{\ell=1}^{d-1} \log G(1+\tfrac{\ell}{d}) +O(\frac{1}{n}). 
\end{split}
\end{align}
This computation again reveals the same coefficient in the $O(\log n)$ term, thereby reinforcing the observed duality between vanishing and diverging densities in the conformal setting.
\end{rem}

Combining Corollary~\ref{Cor_free energy} with the preceding two remarks, we are led to propose an extended version of the conjecture of Jancovici et al. \cite{JMP94, TF99} in the presence of singularities. Specifically, consider the case where the equilibrium measure exhibits a singularity at a single point \( p \in S_W \), whose vanishing or diverging behaviour is of the form \( O(|z - p|^{2d - 2}) \) or \( O(|z - p|^{2/d - 2}) \) for some \( d \in \mathbb{N} \).
Then, the coefficient of the \( \log n \) term in the free energy expansion is conjectured to be
\begin{equation}
\frac{6 - \chi}{12} + \frac{(d - 1)^2}{12d},
\end{equation}
where \( \chi \) denotes the Euler characteristic of the droplet. 
In the regular case \( d = 1 \), this reduces to the original conjecture of Jancovici et al.

\bigskip 

We conclude this section by discussing the special case \( d = 1 \) and its connection to non-Hermitian random matrix theory.
In fact, the lemniscate ensemble in the case \( d = 1 \), where no additional discrete rotational symmetry is imposed, corresponds—up to a trivial translation—to the classical non-Hermitian random matrix ensemble known as the complex Ginibre ensemble, see~\cite{BF25} for a recent review. This ensemble consists of \( n \times n \) random matrices with i.i.d.\ complex Gaussian entries of mean zero and variance \( 1/n \). The joint distribution of their eigenvalues is given by the Coulomb gas form \eqref{def of Gibbs} with the Gaussian potential \( V^{\mathrm{g}}(z) := |z|^2 \). Observe that \( V_{1,t}(z) = V^{\mathrm{g}}(t - z) \). In this case, it is well known that the partition function admits an explicit evaluation
\begin{equation} \label{def of ZN Ginibre}
 Z_{N}^{\mathrm{Gin}} := Z_N( V^{ \rm g } ) =     N^{ -\frac{N(N+1)}{2} } \prod_{k=1}^{N}k! =    N^{ -\frac{N(N+1)}{2} }  \,G(N+2).
\end{equation} 

In the special case \( d = 1 \), the partition function \( Z_n(V_{d,t}^{(c)}) \) is equivalent--up to a normalisation--to the averaged moment of the characteristic polynomial of the complex Ginibre ensemble. 
We now present the result for \( d = 1 \) from the perspective of non-Hermitian random matrix theory.

\begin{thm}[\textbf{Moments of characteristic polynomials for the complex Ginibre matrix}] \label{Prop_moment GinUE}
Fix \( \gamma > -2 \), and let \( a \in \mathbb{C} \). Let \( G_N \) denote the complex Ginibre matrix.
\begin{itemize}
    \item[\textup{(i)}] Let $|a|>1$ be fixed. Then as $N \to \infty$, we have 
\begin{equation} \label{asymp of moments of GinUE}
\log \mathbb{E}\Big( |\det(G_{N}-a)|^{\gamma} \Big) = (\gamma \log |a|) \,N -\frac{ \gamma^2 }{4}  \log\Big( \frac{|a|^2-1}{|a|^2} \Big) +O( \frac{1}{N} ).  
\end{equation} 
      \item[\textup{(ii)}] Let $|a|<1$ be fixed. Then as $N \to \infty$, we have 
\begin{equation} \label{asymp of moments of GinUE bulk}
\log \mathbb{E}\Big( |\det(G_{N}-a)|^{\gamma} \Big) =  \frac{\gamma}{2} (|a|^2-1) N + \frac{\gamma^2}{8} \log N + \frac{\gamma}{4} \log(2\pi) - \log G(1+\tfrac{\gamma}{2})  +O(\frac{1}{N}),
\end{equation}
where $G$ is the Barnes $G$-function.
Furthermore, the $O(N^{-1})$-term in \eqref{asymp of moments of GinUE bulk} is given by $\sum_{m=1}^\infty \mathcal{C}_m N^{-m}$, where 
\begin{align}
\begin{split} \label{def of mathcal Cm}
\mathcal{C}_m & := (-1)^{m+1}\Big( \frac{ \gamma^2 }{4m(m+1)(m+2)}-\frac{1}{12m} \Big) \Big( \frac{\gamma}{2}\Big)^m
\\
&\quad + (-1)^m \bigg( \sum_{k=1}^{ \lfloor(m-1)/2\rfloor } \frac{B_{2k+2}}{ 4k(k+1)(2k-1)! } \frac{(m-1)!}{(m-2k)!} \Big( \frac{\gamma}{2} \Big)^{m-2k} \bigg). 
\end{split}
\end{align}
Here, the second line in \eqref{def of mathcal Cm} is understood to be zero for \( m = 1, 2 \), and \( B_k \) is the \( k \)-th Bernoulli number defined in terms of the generating function as 
\begin{equation} \label{def of Bernoulli number}
\frac{t}{e^t-1} = \sum_{k=0}^\infty B_k \frac{t^k}{k!}.
\end{equation}
\end{itemize} 
\end{thm}

The first few terms of \eqref{def of mathcal Cm} are given as follows.
\begin{align} 
&\mathcal{C}_1 = \frac{ \gamma (\gamma^2-2) }{48},
& & \hspace{-6em} 
\mathcal{C}_2 =  - \frac{\gamma^2(\gamma^2-4)}{ 384 }, 
\\
&\mathcal{C}_3 = \frac{\gamma  (24 - 20 \gamma^2 + 3 \gamma^4 )}{5760 }, 
& & \hspace{-6em}  \mathcal{C}_4 = -\frac{\gamma^2 (24 - 10 \gamma^2 + \gamma^4)}{7680}.  
\end{align}
The asymptotic formula \eqref{asymp of moments of GinUE bulk} in the regime \( |a| < 1 \) was previously obtained by Webb and Wong in~\cite[Theorem~1.1]{WW19}, with an error term of order \( o(1) \). In this respect, our expansion \eqref{asymp of moments of GinUE bulk} provides a refinement of~\cite[Theorem~1.1]{WW19} by supplying the full asymptotic expansion in powers of $1/N$, see also Remark~\ref{Rem_full expansion}. 
Beyond this regime, the moments of characteristic polynomials for complex Ginibre matrices have been investigated in various settings: in \cite{DS22} the edge regime \( |a| \approx 1 \) was studied under fixed exponent; whereas in \cite{BSY24}, the focus lies on the exponentially varying regime, where the exponent scales as \( O(N) \). Moreover, recent extensions of such results have appeared in the context of induced spherical ensembles \cite{BFL25,BFKL25} and truncated unitary ensembles \cite{DMMS25}.

In addition, we note that since Theorem~\ref{Prop_moment GinUE} holds for general real-valued $\gamma >-2$, it directly implies a central limit theorem for the properly normalised logarithm of characteristic polynomials, cf. \cite[Corollary 1.2]{WW19} and \cite[Corollary 1.6]{DMMS25}.

\begin{rem}[Duality for the even integer powers]
For $\gamma = 2k$ with $k$ an integer, Theorem~\ref{Prop_moment GinUE} (i) can be derived via a duality (see \cite{Fo25} for a recent review). More precisely, in this case, the left-hand side of \eqref{asymp of moments of GinUE} admits a representation as a $k$-fold real integral, to which one can apply the standard steepest descent method to extract the asymptotic behaviour, see \cite[(A.7)]{DS22}. 
\end{rem}

\begin{rem}[Unified asymptotic formula in terms of the Barnes $G$-function] \label{Rem_full expansion}
In fact, during the proof of \eqref{asymp of moments of GinUE bulk}, we establish  the asymptotic expansion
\begin{equation} \label{asymp of moments of GinUE bulk Barnes G}
\log \mathbb{E}\Big( |\det(G_{N}-a)|^{\gamma} \Big)  = \log \bigg(  \frac{1}{ N^{ \gamma N/2  }  }   \frac{ G(N+1+\frac{\gamma}{2})  }{ G(N+1) G(  1+\frac{\gamma}{2}  ) } \bigg) + \frac{\gamma}{2}|a|^2 N + O( \frac{1}{N^\infty} ), 
\end{equation}
as \( N \to \infty \), where \( O(1/N^\infty) \) denotes an error term of order \( O(1/N^k) \) for any positive integer \( k \).  
From \eqref{asymp of moments of GinUE bulk Barnes G} and the asymptotic formula \eqref{Barnes G asymp}, the explicit full asymptotic expansion follows.  
Moreover, the asymptotic form \eqref{asymp of moments of GinUE bulk Barnes G} is believed to remain valid in a broader setting, for instance when \( c \) depends on \( N \), see \cite[Remark 2.8]{BSY24} for a related discussion.
\end{rem}

\subsection*{Strategy of the proofs} For the reader’s convenience, we outline the overall strategy of the proofs. The first elementary yet crucial idea in the proof of Theorem~\ref{Thm_free energy expansion} is to eliminate the discrete rotational symmetry, thereby reducing the partition function \( Z_n(V_{d,t}^{(c)}) \) essentially to the case \( d = 1 \). More precisely, in Lemma~\ref{Lem_multifold}, for \( m \in \{0,1,\dots, d-1\} \), we establish the decomposition
\[
\log Z_{dN+m}(V_{d,t}^{(c)}) = \mathcal{A}_1 + \mathcal{A}_2 + \mathcal{A}_3.
\]
Here, the first term \( \mathcal{A}_1 \) is a normalisation constant, expressed in terms of the partition function of the Ginibre ensemble. The second term \( \mathcal{A}_2 \) represents the main contribution, corresponding to the sum of moments of characteristic polynomials of the Ginibre ensemble. The third term \( \mathcal{A}_3 \) consists of the remaining oscillatory contributions, which can be expressed in terms of the orthogonal norms associated with the highest-degree polynomials.

The remaining steps involve analysing the asymptotic behaviour of each term \( \mathcal{A}_1 \), \( \mathcal{A}_2 \), and \( \mathcal{A}_3 \). The first term \( \mathcal{A}_1 \) is the most elementary and follows from the well-known asymptotic behaviour of the Barnes \( G \)-function, see Lemma~\ref{Lem_asymp pre factors}. The third term \( \mathcal{A}_3 \) is derived from the asymptotics of the orthogonal norms, based on recent developments in the associated Riemann--Hilbert analysis \cite{LY17,BY23}, see Lemma~\ref{Lem_asymp of mathcal A3}. The main term, \( \mathcal{A}_2 \), is treated using Theorem~\ref{Prop_moment GinUE}. To prove this theorem, we employ the idea of deforming partition functions developed in recent work \cite{BSY24}. The key idea is to use the rotationally symmetric cases \( a = 0 \) and \( a = \infty \), for which explicit closed-form expressions are available, as reference ensembles. By comparing with these, we analyse the deformation factor via the fine asymptotic behaviour of the orthogonal polynomials near infinity.

We emphasise that a notable feature in the asymptotic analysis is the presence of oscillatory terms. In Theorem~\ref{Thm_free energy expansion}, we observe that the final formula includes an oscillatory contribution in the case \( t > t_c \), while no such term appears when \( t < t_c \). In fact, each of the terms \( \mathcal{A}_1 \), \( \mathcal{A}_2 \), and \( \mathcal{A}_3 \) contains oscillatory components, which in our setting depend on \( m \), the residue of \( n \) modulo \( d \). Remarkably, an interesting computation reveals that for \( t < t_c \), all such oscillatory terms cancel out. In contrast, for \( t > t_c \), the oscillatory contribution survives at the \( O(1) \) level, resulting in the oscillations observed in the final expression. See Table~\ref{Table_asymptotic summary} for a summary.

 \begin{table}[h!]
 	\begin{center}
 		{\def\arraystretch{2}
 			\begin{tabular}{|cc||c|c||}
 				\hline
 				\multicolumn{2}{|c||}{  } &  $t>t_c$  &   $t<t_c$ 
 				\\
 				\hline 
 				\multicolumn{2}{|c||}{ \,\, $\mathcal{A}_1$ \,\, } &   \multicolumn{2}{c||}{ $  m \Big( N + \frac{ \log N }{ 2 }-\frac{\log(2\pi)}{2} +\frac{d-2c-1}{2d} \Big)  $  }          
                \\
 					\hline 
                    	\multicolumn{2}{|c||}{\,\, $\mathcal{A}_2$ \,\,  }  & \quad $   -m\frac{d-2c-1}{2d} +  m\frac{d-2c-1}{d}   \log(\sqrt{d}t)  $  &  \hspace{3.25em}  $   -m\frac{d-2c-1}{2d} $   \hspace{3.25em}  \\
 				\hline  
                    	\multicolumn{2}{|c||}{ \,\, $\mathcal{A}_3$ \,\,  }  &      $  -m \, \Big( N + \frac{ \log N }{ 2 }-\frac{\log(2\pi)}{2} \Big) +\frac{ m(m+1+2c-2d) }{d} \log( \sqrt{d}t )  $   &    \hspace{4em}   $ -m \, \Big( N + \frac{ \log N }{ 2 }-\frac{\log(2\pi)}{2} \Big)    $   \hspace{4em}     \\
 				\hline 
                    	\multicolumn{2}{|c||}{ $\displaystyle \sum_{k=1}^3 \mathcal{A}_k$ }  &  $ \frac{m(m-d)}{d} \log(\sqrt{d}t) $  &  \hspace{3.25em}  $0$  \hspace{3.25em}  \\
 				\hline  
 			\end{tabular}
 		}
 	\end{center}
 	\caption{Oscillatory (\(m\)-dependent) terms in the asymptotic expansions of \(\mathcal{A}_k\) for \(k = 1, 2, 3\) given in Lemmas~\ref{Lem_asymp pre factors}, \ref{Lem_asymp of mathcal A2} and \ref{Lem_asymp of mathcal A3}, respectively.} \label{Table_asymptotic summary}
 \end{table}

\section{Proofs of main results}

In this section, we prove the results in the previous section. We begin with Lemma~\ref{Lem_multifold}, which introduces a multi-fold transform that relates the general case \( d \in \mathbb{N} \) to the simpler case \( d = 1 \). Next, we establish Theorem~\ref{Prop_moment GinUE} by building on recent developments in the theory of associated planar orthogonal polynomials \cite{BY23,BSY24,LY17}. Theorem~\ref{Thm_free energy expansion} then follows by combining Lemma~\ref{Lem_multifold} and Theorem~\ref{Prop_moment GinUE} along with additional analysis on the squared orthogonal norms. Finally, we prove Corollary~\ref{Cor_free energy} by verifying that the coefficients in Theorem~\ref{Thm_free energy expansion} match the conjectural form given in the expansion~\eqref{Z expansion}. Specifically, we confirm that the explicitly computed coefficients $C_1$ and $C_3$ correspond to the expressions involving the logarithmic energy and entropy, cf. \eqref{C1 energy}.

\medskip 

We begin by recalling some well-known facts. Due to the special structure at \( \beta = 2 \), where the pairwise interaction in \eqref{def of Gibbs} yields the squared Vandermonde determinant, the partition function \( Z_n(W) \) admits a determinantal representation:
\begin{equation} \label{Zn in terms of det}
Z_n(W) = n! \,\det \bigg( \int_{\mathbb{C}} \mathcal{P}_j(z) \overline{\mathcal{P}_k(z)}\, e^{-nW(z)} \,\frac{d^2z}{\pi} \bigg)_{j,k=0}^{n-1},
\end{equation}
where \( \mathcal{P}_j \) denotes a monic polynomial of degree \( j \). In particular, if the \( \mathcal{P}_j \) are chosen to be orthogonal with respect to the weight \( e^{-nW(z)}\,\frac{d^2z}{\pi} \), then the partition function reduces to a product of squared orthogonal norms.

Recall that $V_{d,t}^{(c)}$ is given by \eqref{def of Vc lemniscate}.
Let $\mathsf{P}_k \equiv \mathsf{P}_{k,d}^{(c)}$ be the \textit{monic} orthogonal polynomial with respect to $e^{-n V_{d,t}^{(c)}}$, i.e. 
\begin{align} \label{def of orthogonality d gen}
\int_\C \mathsf{P}_j(z) \overline{ \mathsf{P}_k(z) } |z|^{2c} e^{ -n |z^d-t|^2 + nt^2 } \,\frac{d^2z}{\pi} = \mathsf{h}_{j,d}^{(c)}(t)\,\delta_{j,k}, 
\end{align} 
where $\mathsf{h}_{j,d}^{(c)}(t)$ is the squared norm. 
It is also convenient to consider the monic planar orthogonal polynomial $P_j \equiv P_j^{(c)}(z;a)$ satisfying the orthogonality  
\begin{equation} \label{def of orthogonality d 1}
\int_\C P_j(z) \overline{ P_k(z) } |z-a|^{2c}\,e^{-N|z|^2}\,\frac{d^2z}{\pi} = h_{j}^{ (c) }(a) \, \delta_{j,k}. 
\end{equation}

The following lemma establishes a connection between the partition function $Z_n(V_{d,t}^{(c)})$ and the moments of the characteristic polynomials of the Ginibre matrix. It provides an extension of \cite[Lemma 3.6]{DS22}, which treats the special case $c = 0$ and $n = dN$. The underlying idea—based on a multi-fold transform—can be traced back to earlier works \cite{BGM17,BEG18}.

\begin{lem} \label{Lem_multifold}
For $d \in \mathbb{N}$, let $m \in \{0,1,\dots,d-1\}$. Define 
\begin{align}
\mathcal{A}_1 & := \log \Big( e^{(dN+m)^2 t^2}\, c_{N,d}(m) \, (Z_{N}^{\mathrm{Gin}})^d  \Big),  \label{def of mathcal A1}
\\
\mathcal{A}_2 &:= \sum_{\ell=0}^{d-1} \log \mathbb{E}\Big( |\det(G_{N}-a)|^{\gamma_{\ell,c}} \Big),   \label{def of mathcal A2}
\\
\mathcal{A}_3 &:=\sum_{ \ell=0 }^{m-1} \log h_{ N }^{ ( \gamma_{\ell/2} ) }(a),  \label{def of mathcal A3}
\end{align}
where 
\begin{equation} \label{def of cNd m}
c_{N,d}(m) = \frac{ (dN+m)! }{  (N!)^{ d} }  \frac{ 1 }{d^{ dN+m }} \Big( \frac{N}{dN+m}\Big)^{  \frac{d}{2}N^2+ \frac{ 1+2c+2m }{ 2 } N+ \frac{ m(1+2c+m) }{ 2d }  },
\end{equation}
$Z_{N}^{\mathrm{Gin}}$ is given by \eqref{def of ZN Ginibre}, $G_N$ is the $N \times N$ complex Ginibre matrix and 
\begin{equation} \label{def of a and gamma lc}
a= \sqrt{ \frac{Nd+m}{N} }\,t , \qquad \gamma_{ \ell,c }=   -2 \Big( 1- \frac{ \ell+1+c }{d} \Big) . 
\end{equation} 
Then with $n=dN+m$, we have 
\begin{equation} \label{decomposition of log Zn}
\log Z_{n}(V_{d,t}^{(c)}) = \mathcal{A}_1+   \mathcal{A}_2+  \mathcal{A}_3. 
\end{equation} 
\end{lem}

\begin{proof}
Recall that $\mathsf{h}_{j,d}^{(c)}(t)$ and $h_{j}^{ (c) }(a)$ are given by \eqref{def of orthogonality d gen} and \eqref{def of orthogonality d 1}.
We claim that for $\ell=0,1,\dots,d-1$, 
\begin{equation} \label{relation of OP norms}
\mathsf{h}_{dj+\ell,d}^{(c)}(t)= \frac{ e^{(dN+m) t^2} }{d} \Big( \frac{N}{dN+m}\Big)^{ \frac{c+\ell+1}{d} +j  }  h_{j}^{ (\gamma_{ \ell,c }/2) }(a).
\end{equation}
We follow the idea in \cite[Section 3]{BEG18}. Due to the discrete rotational symmetry, we have 
$$
\mathsf{P}_{dj+\ell}(z)=z^\ell\, p_{j}(z^d)  
$$
for some monic polynomial $p_j$. 
Then by the change of variables, 
\begin{align*}
\int_\C \mathsf{P}_{dj+\ell}(z) \overline{ \mathsf{P}_{dk+\ell}(z) } |z|^{2c} e^{ -n |z^d-t|^2 + nt^2 } \,\frac{d^2z}{\pi}  &= d \int_{ 0<\arg z <  \frac{2\pi}{d} } p_j(z^d)\overline{ p_k(z^d) } |z|^{2c+2\ell} e^{ -n |z^d-t|^2 + nt^2 } \,\frac{d^2z}{\pi}
\\
&= \frac{ e^{n t^2} }{d} \int_\C p_j(w) \overline{p_k(w)} |w|^{ \gamma_{ \ell,c } } e^{ -n|w-t|^2 } \frac{d^2w}{\pi}. 
\end{align*}
Note here that with $n=dN+m$, 
\begin{align*}
&\quad \int_\C p_j(w) \overline{p_k(w)} |w|^{ \gamma_{ \ell,c } } e^{ -n|w-t|^2 } \frac{d^2w}{\pi}  = \int_\C p_j(t-w) \overline{p_k(t-w)} |w-t|^{ \gamma_{ \ell,c } } e^{ -n|w|^2 } \frac{d^2w}{\pi}
\\
& =\Big( \frac{N}{dN+m}\Big)^{ \frac{c+\ell+1}{d}  } \int_\C  p_j\Big(t- \sqrt{ \tfrac{N}{dN+m} }  w \Big) \overline{p_k\Big(t- \sqrt{ \tfrac{N}{dN+m} }  w \Big)} \, |w-a|^{ \gamma_{ \ell,c } } e^{ -N |w|^2 } \frac{ d^2 w }{\pi}. 
\end{align*}
Therefore, it follows from \eqref{def of orthogonality d gen} and \eqref{def of orthogonality d 1} that 
\begin{equation}
p_j\Big(t- \sqrt{ \tfrac{N}{dN+m} }  z \Big) = \Big( -\tfrac{N}{dN+m}\Big)^{j/2} P_j^{ (\gamma_{\ell,c}/2) }(z).
\end{equation}
By using this, we obtain 
\begin{align}
\begin{split}
&\quad \int_\C |\mathsf{P}_{dj+\ell}(z)|^2\, |z|^{2c} e^{ -n |z^d-t|^d + nt^d } \,\frac{d^2z}{\pi}  
\\
&=  \frac{ e^{(dN+m) t^2} }{d} \Big( \frac{N}{dN+m}\Big)^{ \frac{c+\ell+1}{d} +j  }  \int_\C  | P_j^{ (\gamma_{\ell,c}/2) }(w)  |^2 |w-a|^{ \gamma_{ \ell,c } } e^{ -N |w|^2 } \frac{ d^2 w }{\pi},
\end{split}
\end{align}  
which leads to \eqref{relation of OP norms}. 

Note that by \eqref{Zn in terms of det} and \eqref{def of orthogonality d gen}, we have 
\begin{equation}
Z_{dN+m}(V_{d,t}^{(c)}) = (dN+m)! \prod_{j=0}^{dN+m-1}  \mathsf{h}_{j,d}^{(c)}(t) .
\end{equation}
On the other hand, by letting 
$$Q(z):=|z|^2-\frac{\gamma_{\ell,c}}{n} \log|z-a|,$$ 
we have   
\begin{equation} 
\mathbb{E}\Big( |\det(G_{N}-a)|^{\gamma_{\ell,c}} \Big)  = \frac{ Z_N( Q )  }{ Z_N^{ \rm Gin } } = \frac{N!}{ Z_N^{ \rm Gin } } \prod_{j=0}^{N-1} h_{ N }^{ ( \gamma_{\ell,c}/2 ) }(a). 
\end{equation}
Combining these with \eqref{relation of OP norms}, we obtain
\begin{align*}
Z_{dN+m}(V_{d,t}^{(c)}) 
&= \frac{(dN+m)!}{ (N!)^d } \bigg( \prod_{\ell=0}^{d-1}  N! \prod_{j=0}^{N-1}  \mathsf{h}_{dj+\ell,d}^{(c)} (t)  \bigg) \prod_{j=dN}^{dN+m-1}  \mathsf{h}_{j,d}^{(c)} (t) 
\\
&= e^{(dN+m)^2 t^2}\, c_{N,d}(m) \bigg( \prod_{\ell=0}^{d-1}  N! \prod_{j=0}^{N-1}  h_{j}^{( \gamma_{\ell,c}/2 )} (a) \bigg) \prod_{ \ell=0 }^{m-1} h_{ N }^{ ( \gamma_{\ell,c}/2 ) }(a)
\\
&= e^{(dN+m)^2 t^2}\, c_{N,d}(m) \, (Z_{N}^{\mathrm{Gin}})^d  \prod_{\ell=0}^{d-1}\mathbb{E}\Big( |\det(G_{N}-a)|^{\gamma_{\ell,c}} \Big) \prod_{ \ell=0 }^{m-1} h_{ N }^{ ( \gamma_{\ell/2} ) }(a), 
\end{align*}
which completes the proof.  
\end{proof}

We compute the asymptotic behaviour of $\mathcal{A}_1$.

\begin{lem} \label{Lem_asymp pre factors}
For $d \in \mathbb{N}$, let $m \in \{0,1,\dots,d-1\}$. Let $n=dN+m$ and $\mathcal{A}_1$ be given by \eqref{def of mathcal A1}. Then as $n \to \infty$, we have 
\begin{align}
\begin{split} \label{final sum 1}
\mathcal{A}_1&= \Big( t^2 -\frac{3}{4d}-\frac{\log d}{2 d} \Big) n^2+ \frac12 n\log n
 +  \Big( \frac{\log(2\pi)}{2} -1 -\frac{  1+2c+d }{2d} \log d  \Big) n + \frac{6-d}{12} \log n
\\
&\quad    + d \,\zeta'(-1) + \frac{d}{12} \log d + \frac{\log(2\pi)}{2}  +m \Big( N + \frac{ \log N }{ 2 }-\frac{\log(2\pi)}{2} +\frac{d-2c-1}{2d} \Big)  +O(\frac{1}{n}).
\end{split}
\end{align}  
\end{lem}
\begin{proof}
Recall the asymptotic behaviours of the gamma and Barnes $G$-function \cite[Eqs.~(5.11.1), (5.17.5)]{NIST}: as $ z \to +\infty,$   
\begin{align} 
\label{gamma asymp}
 \log \Gamma(z)&= \Big(z-\frac12\Big) \log z - z +\frac12 \log(2\pi) +\sum_{k=1}^\infty \frac{B_{2k}}{ 2k(2k-1)z^{2k-1} }, \\
\begin{split} \label{Barnes G asymp}
\log G(z+1) & =\frac{z^2 \log z}{2} -\frac34 z^2+\frac{ \log(2\pi) z}{2}-\frac{\log z}{12}+\zeta'(-1)  + \sum_{k=1}^\infty \frac{ B_{2k+2} }{ 4k(k+1)  } \frac{1}{z^{2k}},
\end{split}
\end{align}
where $B_k$ is the Bernoulli number \eqref{def of Bernoulli number}. 
Then by straightforward computations using \eqref{gamma asymp} and \eqref{Barnes G asymp}, the lemma follows.  
\end{proof}

Next, we prove Theorem~\ref{Prop_moment GinUE}. 

\begin{proof}[Proof of Theorem~\ref{Prop_moment GinUE}]
Due to the rotational invariance, it suffices to consider the case $a \ge 0$. 
One of the key ingredients is the following differential identity \cite[Proposition 3.9]{BSY24}, derived in the context of \(\tau\)-function theory:
\begin{equation} \label{eqn for tau derivative}
\frac{d}{da} \log \mathbb{E}\Big( |\det(G_{N}-a)|^{2c} \Big) = 2N \mathfrak{A},
\end{equation}
where $\mathfrak{A}$ is defined by 
\begin{equation} \label{def of mathfrak A}
  \mathfrak{A}= \lim_{z \to \infty} \frac{P_N^{(c)}(z) - z^N}{z^{N-1}}. 
\end{equation}
The remaining step is then to compute \( \mathfrak{A} \) in each regime.

We first show \eqref{asymp of moments of GinUE}. Suppose that $a>1$. 
By \cite[Proposition 2.3]{BY23}, for $a>1$, $c>-1$ and for $z$ outside  the unit disc\footnote{Indeed, the asymptotic formula is valid for \( z \) lying outside the Szeg\H{o} curve, which itself is contained within the unit disc. For our purposes, we only require the behaviour as \( z \to \infty \), and therefore omit the precise definition of the Szeg\H{o} curve.}, we have 
\begin{equation}
P_N^{ (c) }(z) = z^{N}\Big(\frac{z}{z-1/a}\Big)^{c} \bigg[ 1- \frac{c}{1-a z}\Big( \frac{1+c}{2}\frac{1}{1-a z }+\frac{c}{ 1-a^2 } \Big) \frac{1}{N} +O( \frac{1}{N^2} )  \bigg],
\end{equation}
as $N \to \infty$. This in turns implies that 
\begin{equation}
P_N^{ (c) }(z)= z^N+ \frac{c}{a}\Big( 1 +\frac{c}{ 1-a^2 }\frac{1}{N} \Big) z^{N-1}+O(z^{N-2}), \qquad \textup{as }z \to \infty. 
\end{equation}
Therefore by \eqref{eqn for tau derivative}, we obtain 
\begin{equation} \label{derivative of log characteristic}
\frac{d}{da} \log \mathbb{E}\Big( |\det(G_{N}-a)|^{2c} \Big) =  \frac{2c}{a} N+\frac{2c^2}{a(1-a^2)} +O(\frac{1}{N}) .
\end{equation}
Note that since 
\begin{align*}
\mathbb{E}\Big( |\det(G_{N}-a)|^{2c} \Big)& = \frac{1}{Z_N^{ \rm Gin }} \int_{ \C^N } \prod_{1 \le j<k \le N} |z_{j}-z_{k}|^{2} \prod_{j=1}^{N}|z_j-a|^{2c} e^{-N|z_j|^2}  \frac{d^{2}z_{j}}{\pi}
\\
&  =  \frac{ a^{2cN} }{Z_N^{ \rm Gin }} \int_{ \C^N } \prod_{1 \le j<k \le N} |z_{j}-z_{k}|^{2} \prod_{j=1}^{N}|1-z_j/a|^{2c} e^{-N|z_j|^2}  \frac{d^{2}z_{j}}{\pi},
\end{align*}
we have 
\begin{equation}
\lim_{a \to \infty} \bigg( \log \mathbb{E}\Big( |\det(G_{N}-a)|^{2c} \Big) - (2c \log a )N \bigg)=0. 
\end{equation}
(See also \cite[Lemma 3.1]{BSY24}, which states essentially the same formula up to normalisation.)
Therefore by integrating \eqref{derivative of log characteristic}, we obtain \eqref{asymp of moments of GinUE}. 

Next, we show \eqref{asymp of moments of GinUE bulk}. Suppose that $|a|<1$. By \cite[Theorem 3]{LY17}, for $|a|<1$, $c>-1$ and for $z$ outside  the unit disc, we have 
\begin{equation}
P_N^{(c)}(z) = z^N \Big( \frac{z}{z-a}\Big)^c \Big( 1+O(\frac{1}{N^\infty})\Big),
\end{equation}
as $N \to \infty$. Therefore, we have 
\begin{equation}
P_N^{ (c) }(z)= z^N+ \Big( ac +O(\frac{1}{N^\infty}) \Big) z^{N-1}+O(z^{N-2}), \qquad \textup{as }z \to \infty. 
\end{equation}
Then it follows that 
\begin{equation}
\frac{d}{da} \log \mathbb{E}\Big( |\det(G_{N}-a)|^{2c} \Big) = 2 ac N+O(\frac{1}{N^\infty}). 
\end{equation}
On the other hand, for the case $a=0$, by \eqref{Zn t0 Barnes G} with $d=1$ and \eqref{def of ZN Ginibre},
\begin{equation}
\begin{split}
\mathbb{E}\Big( |\det(G_{N})|^{2c} \Big) &= \frac{1}{Z_N^{ \rm Gin }}   \frac{N!}{ N^{ \frac{N(N+2c+1)}{2} }  }   \frac{ G(N+c+1)  }{ G(   c+1  ) }  
=   \frac{1}{ N^{ cN  }  }   \frac{ G(N+c+1)  }{ G(N+1) G(   c+1  ) }  . 
\end{split}
\end{equation}
Combining these with \eqref{Barnes G asymp}, we obtain \eqref{asymp of moments of GinUE bulk}. In particular, we have proven \eqref{asymp of moments of GinUE bulk Barnes G}. Furthermore, by \eqref{Barnes G asymp}, the error term $O(N^{-1})$ is given by  
$$
\sum_{k=1}^\infty(-1)^{k+1}\Big( \frac{ \gamma^2 }{4k(k+1)(k+2)}-\frac{1}{12k} \Big) \Big( \frac{\gamma}{2}\Big)^k \frac{1}{N^k} + \sum_{k=1}^\infty \frac{ B_{2k+2} }{ 4k(k+1) } \bigg( \frac{1}{(N+\frac{\gamma}{2})^{2k} }-\frac{1}{N^{2k}}\bigg). 
$$
Then by straightforward computations, we obtain \eqref{def of mathcal Cm}. 
\end{proof}

As an immediate consequence of Theorem~\ref{Prop_moment GinUE}, we obtain asymptotic behaviour of $\mathcal{A}_2$ in \eqref{def of mathcal A2}. 

\begin{lem} \label{Lem_asymp of mathcal A2}
For $d \in \mathbb{N}$, let $m \in \{0,1,\dots,d-1\}$. Let $n=dN+m$ and $\mathcal{A}_2$ be given by \eqref{def of mathcal A2}. Then as $n \to \infty$, we have the following.
\begin{itemize}
    \item[\textup{(i)}] Let $t>t_c$. Then we have 
    \begin{align}
\begin{split} \label{final sum 2}
\mathcal{A}_2   &= \frac{ (1+2c-d) \log( \sqrt{d} t  ) }{ d } n + \Big( \frac{c(d-c-1)}{d}-\frac{(d-1)(2d-1)}{6d} \Big) \log\Big( \frac{dt^2-1}{dt^2} \Big)
\\
&\quad + m  \frac{1+2c-d}{2d} \Big( 1-2 \log(\sqrt{d}t)  \Big)   +O( \frac{1}{n} ) .  
\end{split}
\end{align}
\item[\textup{(ii)}] Let $t<t_c$. Then  we have
\begin{align}
\begin{split} \label{final sum 2 bulk}
\mathcal{A}_2   &= \frac{1+2c-d}{2d}(dt^2-1) n +  \Big( \frac{(d-1)(2d-1)}{12d}-\frac{c(d-c-1)}{2d} \Big) \log \Big(\frac{n}{d}\Big)
\\
&\quad + \frac{ 1+2c-d }{4 }\log(2\pi) -\sum_{\ell=0}^{d-1} \log G\Big(  \frac{\ell+1+c}{d} \Big) + m\frac{ 1+2c-d }{2d}  +O(\frac{1}{n}). 
\end{split}
\end{align} 
\end{itemize}
\end{lem} 
\begin{proof}
The lemma follows from straightforward computations using Theorem~\ref{Prop_moment GinUE}, \eqref{def of mathcal A2} and \eqref{def of a and gamma lc}. 
\end{proof}

Finally, we compute asymptotic behaviour of $\mathcal{A}_3$ in \eqref{def of mathcal A3}.

\begin{lem} \label{Lem_asymp of mathcal A3}
For $d \in \mathbb{N}$, let $m \in \{0,1,\dots,d-1\}$. Let $n=dN+m$ and $\mathcal{A}_3$ be given by \eqref{def of mathcal A3}. Then as $n \to \infty$, we have the following.
\begin{itemize}
    \item[\textup{(i)}] Let $t>t_c$. Then we have 
   \begin{align} \label{final sum 3}
\mathcal{A}_3 &= -m \Big( N + \frac{ \log N }{ 2 }-\frac{\log(2\pi)}{2} \Big) +\frac{ m(m+1+2c-2d) }{d} \log( \sqrt{d}t )+O(\frac{1}{n}).
\end{align}
\item[\textup{(ii)}] Let $t<t_c$. Then  we have
\begin{equation} \label{final sum 3 bulk}
\mathcal{A}_3 = -m \Big( N + \frac{ \log N }{ 2 }-\frac{\log(2\pi)}{2} \Big) +O(\frac{1}{n}).
\end{equation}
\end{itemize}
\end{lem}
\begin{proof}
Let $n=dN+m$.  By \cite[Lemma 2.4]{BY23}, for $|a|>1$ and $c>-1$, we have  
\begin{equation}  \label{asymptotic op norm d1 out}
h_N^{(c)} (a)=  e^{-N} \sqrt{ \frac{2\pi}{N} } |a|^{2c} \Big(1+O(\frac{1}{N})\Big),
\end{equation}
as $N \to \infty.$ Using this, \eqref{final sum 3} immediately follows. 
On the other hand, for $|a|<1$, we have
\begin{equation} \label{asymptotic op norm d1 bulk}
h_N^{(c)}(a) = e^{-N} \sqrt{ \frac{2\pi}{N} }  \Big( 1+O(\frac{1}{N}) \Big). 
\end{equation}
This asymptotic behaviour follows from the computations in \cite[Subsection 3.3]{LY23}, combined with the Riemann--Hilbert analysis in \cite{LY17}. For the reader’s convenience, we provide the details in Appendix~\ref{Appendix_norm}. Using \eqref{asymptotic op norm d1 bulk}, the desired asymptotic behaviour \eqref{final sum 3 bulk} follows. 
\end{proof}

\begin{rem}
In the special case \( c = 1 \), the orthogonal norm can be expressed in terms of the incomplete Gamma function \( Q(s, z) \) as  
$$
h_k^{(1)}	= \frac{(k+1)!}{N^{k+2}} \frac{Q(k+2,N a^2)}{ Q(k+1,Na^2) },
$$ 
see \cite[Appendix A]{BY23}. The asymptotic behaviours in \eqref{asymptotic op norm d1 out} and \eqref{asymptotic op norm d1 bulk} can then be rederived using the standard asymptotic expansions of the incomplete Gamma function \cite[Section 8.11]{NIST}. 
\end{rem}

We are now ready to prove Theorems~\ref{Thm_free energy expansion}. 

\begin{proof}[Proof of Theorems~\ref{Thm_free energy expansion}] 
Combining the decomposition \eqref{decomposition of log Zn} from Lemma~\ref{Lem_multifold} with the asymptotic behaviours of each term provided in Lemmas~\ref{Lem_asymp pre factors}, \ref{Lem_asymp of mathcal A2}, and \ref{Lem_asymp of mathcal A3}, respectively, and after straightforward simplifications, the theorems follow.
We note that each of the expressions in Lemmas~\ref{Lem_asymp pre factors}, \ref{Lem_asymp of mathcal A2}, and \ref{Lem_asymp of mathcal A3} contains oscillatory terms depending on \( m \).  
However, the oscillatory terms of orders \( O(N) \) and \( O(\log N) \) cancel out upon summation.  
Furthermore, in the case \( t < t_c \), the oscillatory term of order \( O(1) \) also cancels.  
On the other hand, when \( t > t_c \), an oscillatory contribution of order \( O(1) \) remains after the summation.  
See Table~\ref{Table_asymptotic summary} for a summary of these oscillatory contributions.
\end{proof}

Finally, we prove Corollary~\ref{Cor_free energy}.

\begin{proof}[Proof of Corollary~\ref{Cor_free energy}]

By Theorem~\ref{Thm_free energy expansion}, it remains to verify that for the equilibrium measure \( \sigma_{V_{d,t}} \) defined in \eqref{def of eq msr}, the constants \( C_1 \) and \( C_3 \) in Theorem~\ref{Thm_free energy expansion} satisfy the following identities:
\begin{equation} \label{C1 energy}
C_1= -I_{V_{d,t}}[\sigma_{ V_{d,t} }], \qquad C_3 |_{c=0} =   \frac{\log (2\pi)}{2}-1  -\frac12 \int_{S_{ V_{d,t} } } \log( \Delta V_{d,t}(z) )  \Delta V_{d,t}(z)\,\frac{d^2z}{\pi}.  
\end{equation} 

We first consider a general potential $V(z)$ having discrete rotational symmetry $V(z)=V(e^{2\pi i/d}z)$. 
Let $Q$ be a potential such that
\begin{equation}
V(z)= \frac{1}{d} Q(z^d). 
\end{equation}
Then it follows from \cite[Lemma 1]{BM15} that 
\begin{equation}
S_V= \{ z \in \C: z^d \in S_Q \}. 
\end{equation} 
Due to standard potential theory \cite{ST97}, the equilibrium measure $\sigma_W$ is characterised by the variational conditions  
\begin{equation} \label{eq:variational}
  \int_\mathbb{C} \log\frac{1}{|z-w|  } d\sigma_W(w) + \frac{1}{2}W(z)  \begin{cases}
  = F_W &   z \in S_W, 
    \smallskip   
    \\
\ge F_W &  z \in \mathbb{C}. 
  \end{cases}
\end{equation} 
Then by definition \eqref{def of log energy}, the energy $I_W[\sigma_W]$ can be written as 
\begin{align}
\begin{split} \label{energy Robin}
I_W[\sigma_W] = F_W + \frac12 \int W(z) \,d\sigma_W(z).
\end{split}
\end{align} 
Combining $F_V=F_Q/d$ (see \cite[p.409]{BM15}) and the change of variables 
\begin{align*}
\int_{S_V} V(z)\Delta V(z) \,\frac{d^2z}{\pi}= \frac{1}{d} \int_{ S_Q } Q(z)\Delta Q(z) \,\frac{d^2z}{\pi}, 
\end{align*} 
we have 
\begin{equation} \label{relation energies d 1}
I_V[\sigma_V]= \frac{1}{d} I_Q[\sigma_Q]. 
\end{equation}

Note that for the potential \eqref{def of V lemniscate}, we have 
\begin{equation}
	V_{d,t}(z)= \frac{1}{d} Q_t(z^d), \qquad Q_t(z):=   d|z-t|^2-dt^2. 
\end{equation}
Since the energy associated with the quadratic potential $|z|^2$ is $3/4$ (see e.g. \cite[Eq.~(1.13)]{BKS23}), after trivial translation and dilation, one can see that 
$$
I_{ Q_t }[ \sigma_{ Q_t } ]= \frac34-dt^2+\frac{\log d}{2}.
$$
Then by \eqref{relation energies d 1}, we obtain the first identity in \eqref{C1 energy}. 

Next, we compute $C_3$. By the change of variables $z=w^d$, the entropy can be computed as
\begin{align}
\begin{split}\label{2.41}
&\quad \int_{ S_{ V_{d,t} } } \log( \Delta  V_{d,t}(z) )  \Delta  V_{d,t}(z)\,\frac{d^2z}{\pi} 
= \int_{  S_{ V_{d,t} } } \log(d^2 |z|^{2d-2}) (d^2 |z|^{2d-2})\,\frac{d^2z}{\pi} 
\\ 
&=d \int_{ |w-t|<1/\sqrt{d} } \log(d^2 |w|^{(2d-2)/d} ) \,\frac{d^2w}{\pi} 
= 2\log(d) +2(d-1) \int_{ |w-t|<1/\sqrt{d} } \log |w|  \,\frac{d^2w}{\pi} . 
\end{split}
\end{align}
Note that by Jensen's formula, for $R>0$ and $p \in \mathbb{C}$ we have 
\begin{equation} \label{2.42}
	\int_{ |w-p|<R } \log|z-w| \,\frac{d^2w}{\pi} = 
    \begin{cases}
    R^2 \log|z-p| & \text{if } |z-p|>R, 
    \smallskip 
    \\
    R^2 \log R- \dfrac{R^2}{2} +\dfrac{|z-p|^2}{2} & \textup{if }|z-p|<R, 
    \end{cases} 
\end{equation}
see e.g. \cite[Lemma 2.6]{By24}.
By combining \eqref{2.41} and \eqref{2.42}, we obtain 
\begin{equation}
\int_{ S_{ V_{d,t} } } \log( \Delta  V_{d,t}(z) )  \Delta  V_{d,t}(z)\,\frac{d^2z}{\pi}  
= 
\begin{cases}
\displaystyle  2 \log(d) + \frac{2(d-1)}{d} \log(t),  &\textup{if } t> t_c,
  \smallskip
  \\
\displaystyle   \frac{d-1}{d} (dt^2-1) + \frac{1+d}{d} \log (d) & \textup{if } t< t_c. 
\end{cases} 
\end{equation} 
This gives rise to the second identity in \eqref{C1 energy}.  
\end{proof}

\medskip 

\appendix

\section{Asymptotic of the orthogonal norm} \label{Appendix_norm}

In this appendix, we provide a detailed derivation of \eqref{asymptotic op norm d1 bulk}.  
Recall that \( P_n \equiv P_n^{(c)} \) denotes the orthogonal polynomial and \( h_n \equiv h_n^{(c)} \) 
its squared norm, as defined in \eqref{def of orthogonality d 1}.  
For the latter purpose, we introduce a subscript \( N \) (i.e. \( P_{n,N} \) and \( h_{n,N} \)) to indicate the dependence on the scaling factor \( N \) appearing in the weight function \( e^{-N|z|^2} \) in \eqref{def of orthogonality d 1}.
This dependence amounts to a simple change of scaling, as expressed by the relation
\begin{equation}
\label{eq op relation}
P_{n,N}(z;a)=\Big(\frac{n}{N}\Big)^{\frac{n}{2}}P_{n,n}\Big(\sqrt{\tfrac{N}{n}}z,\sqrt{\tfrac{N}{n}}a\Big),
\end{equation}
see e.g. \cite[Eq. (3.37)]{BY23}.

In \cite{BBLM15}, the planar orthogonality \eqref{asymptotic op norm d1 bulk} is reformulated as a contour orthogonality, from which the associated Riemann--Hilbert problem can be constructed.
In particular, by \cite[Proposition 7.1]{BBLM15}, we have 
\begin{equation}\label{eq hn}
    h_n^{(c)}=-\frac{1}{\pi} \frac{\Gamma(c+n+1)}{2iN^{c+n+1}}\frac{\widetilde{h}_n^{(c)}}{P_{n+1,N}^{(c)}(0)}, \qquad \widetilde{h}_n^{(c)} \equiv \widetilde{h}_{n,N}^{(c)}(a):=\int_\Gamma P_{n,N}(z)^2w_{n,N}(z)\, dz,
\end{equation}
where $\Gamma$ is a simple closed curve that encloses the line segment $[0,a]$ and 
\begin{equation} \label{weight contour}
w_{n,N}(z):=\Big(\frac{z-a}{z}\Big)^c\frac{e^{-Naz}}{z^n}. 
\end{equation}
By \eqref{eq op relation}, we have
\begin{equation}
\label{eq hn relation}
\widetilde{h}_{n,N}(a)=\Big(\frac{n}{N}\Big)^{\frac{n+1}{2}}\widetilde{h}_{n,n}\Big(\sqrt{\tfrac{N}{n}}a\Big),
\end{equation}
see \cite[Eq.~(3.41)]{BY23}.

On the other hand, by the second asymptotic behaviour in \cite[Theorem 3]{LY17}, we have  
\begin{equation}
P_{N,N}^{(c)}(0) =  \frac{ a^{N } (1-a^2)^{c-1} }{ N^{1-c}\Gamma(c) }  e^{-Na^2 }  \Big( 1+O(\frac{1}{N}) \Big). 
\end{equation}
Combining this with \eqref{eq op relation}, we have
\begin{align}
\begin{split} \label{eq opN1}
P_{N+1,N}^{(c)}(0) &=\Big(\frac{N+1}{N}\Big)^{\frac{N+1}{2}}P_{N+1,N+1}\Big(0,\sqrt{\tfrac{N}{N+1}}a\Big)
=      \frac{ a^{N+1} (1- a^2)^{c-1} }{ N^{1-c}\Gamma(c) }  e^{-Na^2 }  \Big( 1+O(\frac{1}{N}) \Big). 
\end{split}
\end{align} 

Now, we define the matrix function $Y(z)$ by
\begin{equation}
Y(z):= 
\begin{bmatrix}
P_n(z)&\displaystyle\frac{1}{2\pi i}\int_\Gamma\frac{P_n(s)w_{n,N}(s)}{s-z}\, ds
\smallskip 
\\
   Q_{n-1}(z)&\displaystyle\frac{1}{2\pi i}\int_\Gamma\frac{Q_{n-1}(s)w_{n,N}(s)}{s-z} \, ds
\end{bmatrix},
\end{equation}
where $Q_{n-1}$ is the unique polynomial of degree $n-1$ satisfying
\begin{equation}\label{eq cond qn}
\displaystyle\frac{1}{2\pi i}\int_\Gamma\frac{Q_{n-1}(s)w_{n,N}(s)}{s-z} \, ds =\frac{1}{z^n} \cdot \Big( 1+O(\frac{1}{z}) \Big).
\end{equation}
Then we have 
\begin{equation} \label{norm in terms of Y12}
\widetilde{h}_N^{(c)} =-2\pi i  \lim_{z\to\infty}z^{N+1}[Y(z)]_{12}, 
\end{equation}
see e.g. \cite[Eq. (3.45)]{BY23}.  
Furthermore, by solving the associated Riemann-Hilbert problem, it was shown in \cite[Section 6]{LY17} that for $z$ outside the Szeg\H{o} curve, 
\begin{align}
\label{RHP of Y}
\begin{split}
    Y(z)  =e^{\frac{Nl}{2}\sigma_3}\Big(I+O(\frac{1}{N^\infty})\Big)\, R(z)\Phi(z)e^{-\frac{Nl}{2}\sigma_3}z^{N\sigma_3},
\end{split}
    \end{align}
where $\sigma_3$ is the third Pauli matrix, $l=\log a -a^2$, and 
\begin{equation}
R(z) = I+ \frac{a (1-a^2)^{c-1} }{ N^{1-c} \Gamma(c) } \frac{1}{z-a} \begin{bmatrix}
    0 & 1
    \\
    0 & 0
\end{bmatrix}, \qquad \Phi(z)=  \begin{bmatrix}\big(\frac{z}{z-a}\big)^c&0\\
  0& \big(\frac{z-a}{z}\big)^c\end{bmatrix}. 
\end{equation} 
Then by \eqref{norm in terms of Y12} and \eqref{RHP of Y}, we have 
\begin{align} \label{h tilde asymp}
 \widetilde{h}_N^{(c)} = -2\pi i\,  a^{N+1} e^{-Na^2} \Big( \frac{  (1-a^2)^{c-1}  }{ \Gamma(c) } \frac{ 1 }{ N^{1-c} } +O(\frac{1}{N})\Big). 
\end{align}
Now by combining \eqref{eq hn}, \eqref{eq opN1} and \eqref{h tilde asymp}, the desired asymptotic formula \eqref{asymptotic op norm d1 bulk} follows.

\subsection*{Acknowledgement} The author was supported by the National Research Foundation of Korea grants (RS-2023-00301976, RS-2025-00516909). The author thanks Yacin Ameur, Christophe Charlier, Tamara Grava, Nick Simm and Meng Yang for helpful discussions during the preparation of the manuscript.

\bibliographystyle{abbrv}

\end{document}